\documentclass[12pt]{article}
\usepackage[latin1]{inputenc}
\usepackage[british]{babel}
\usepackage{cmap}
\usepackage{lmodern}

\usepackage{amssymb, amsmath, amsthm}
\usepackage[a4paper,top=25mm,bottom=25mm,left=25mm,right=25mm]{geometry}
\usepackage{etex}
\usepackage{ragged2e}

\usepackage{authblk} 
\usepackage{pifont}
\usepackage{graphicx}
\usepackage[usenames,dvipsnames,svgnames,table]{xcolor}
\usepackage[figuresright]{rotating}
\usepackage{xtab} 
\usepackage{longtable} 
\usepackage{multirow}
\usepackage{footnote}
\usepackage[stable]{footmisc}
\usepackage{chngpage} 
\usepackage{pdflscape} 
\usepackage{tocbibind} 

\usepackage{pgfplots}
\pgfplotsset{compat=1.14}
\usepackage{setspace}

\makesavenoteenv{tabular}
\usepackage{tabularx}
\usepackage{booktabs}
\usepackage{threeparttable}
\usepackage[referable]{threeparttablex} 
\newcolumntype{R}{>{\raggedleft\arraybackslash}X}
\newcolumntype{L}{>{\raggedright\arraybackslash}X}
\newcolumntype{C}{>{\centering\arraybackslash}X}
\newcolumntype{A}{>{\columncolor{gray!25}}C}
\newcolumntype{a}{>{\columncolor{gray!25}}c}

\newlength{\tablen}

\usepackage{dcolumn} 
\newcolumntype{.}{D{.}{.}{-1}}

\usepackage{tikz}
\usetikzlibrary{arrows, calc, matrix, patterns, positioning, trees}
\usepackage[semicolon]{natbib}
\usepackage[hyphens]{url}
\usepackage{hyperref} 
\hypersetup{
  colorlinks   = true,    
  urlcolor     = blue,    
  linkcolor    = blue,    
  citecolor    = red      
}
\usepackage{microtype}
\usepackage[justification=centering]{caption} 

\usepackage[labelformat=simple]{subcaption}

\DeclareCaptionLabelFormat{parenthesis}{(#2)}
\makeatletter
\renewcommand\p@subfigure{\arabic{figure}.}
\makeatother

\DeclareCaptionLabelFormat{parenthesis}{(#2)}
\makeatletter
\renewcommand\p@subtable{\arabic{table}.}
\makeatother

\usepackage{enumitem}

\setlist[itemize]{leftmargin=2.5\parindent}
\setlist[enumerate]{leftmargin=2.5\parindent}

\theoremstyle{plain}

\newtheorem{corollary}{Corollary}[section]
\newtheorem{lemma}{Lemma}[section]
\newtheorem{proposition}{Proposition}[section]
\newtheorem{theorem}{Theorem}[section]

\theoremstyle{definition}
\newtheorem{axiom}{Axiom}[section]

\newtheorem{definition}{Definition}[section]
\newtheorem{example}{Example}[section]

\theoremstyle{remark}
\newtheorem{notation}{Notation}[section]

\def\keywords{\vspace{.5em} 
{\noindent \textit{Keywords}:\,}}

\def\JEL{\vspace{.5em} 
{\noindent \textbf{\emph{JEL} classification number}:\,}}

\def\AMS{\vspace{.5em} 
{\noindent \textbf{\emph{MSC} class}:\,}}

\author{\href{https://sites.google.com/site/laszlocsato87}{L\'aszl\'o Csat\'o}\thanks{~e-mail: laszlo.csato@uni-corvinus.hu} }
\affil{Institute for Computer Science and Control, Hungarian Academy of Sciences (MTA SZTAKI) \\
Laboratory on Engineering and Management Intelligence, Research Group of Operations Research and Decision Systems}
\affil{Corvinus University of Budapest (BCE) \\
Department of Operations Research and Actuarial Sciences}
\affil{Budapest, Hungary}
\title{An impossibility theorem for paired comparisons}
\date{\today}

\begin{document}

\maketitle

\begin{abstract}
In several decision-making problems, alternatives should be ranked on the basis of paired comparisons between them. We present an axiomatic approach for the universal ranking problem with arbitrary preference intensities, incomplete and multiple comparisons. In particular, two basic properties -- independence of irrelevant matches and self-consistency -- are considered. It is revealed that there exists no ranking method satisfying both requirements at the same time. The impossibility result holds under various restrictions on the set of ranking problems, however, it does not emerge in the case of round-robin tournaments. An interesting and more general possibility result is obtained by restricting the domain of independence of irrelevant matches through the concept of macrovertex.

\JEL{C44, D71}

\AMS{15A06, 91B14}

\end{abstract}

\section{Introduction}

Paired-comparison based ranking emerges in many fields of science such as social choice theory \citep{ChebotarevShamis1998a}, sports \citep{Landau1895, Landau1914, Zermelo1929, Radicchi2011, BozokiCsatoTemesi2016, ChaoKouLiPeng2018}, or psychology \citep{Thurstone1927}. Here a general version of the problem, allowing for different preference intensities (including ties) as well as incomplete and multiple comparisons between two objects, is addressed.

The paper contributes to this field by the formulation of an impossibility theorem: it turns out that two axioms, independence of irrelevant matches -- used, among others, in characterizations of Borda ranking by \citet{Rubinstein1980} and \citet{NitzanRubinstein1981} and recently discussed by \citet{Gonzalez-DiazHendrickxLohmann2013} -- and self-consistency -- a less known but intuitive property, introduced in \citet{ChebotarevShamis1997a} -- cannot be satisfied at the same time.
We also investigate domain restrictions and the weakening of the properties in order to get some positive results.

Our main theorem reinforces that while the row sum (sometimes called Borda or score) ranking has favourable properties in the case of round-robin tournaments, its application can be attacked when incomplete comparisons are present. A basket case is a Swiss-system tournament, where row sum seems to be a bad choice since players with weaker opponents can score the same number of points more easily \citep{Csato2013a, Csato2017c}.

The current paper can be regarded as a supplement to the findings of previous axiomatic discussions in the field \citep{AltmanTennenholtz2008, ChebotarevShamis1998a, Gonzalez-DiazHendrickxLohmann2013, Csato2018g} by highlighting some unknown connections among certain axioms.
Furthermore, our impossibility result gives mathematical justification for a comment appearing in the axiomatic analysis of scoring procedures by \citet{Gonzalez-DiazHendrickxLohmann2013}: 'when players have different opponents (or face opponents with different intensities), $IIM$\footnote{~$IIM$ is the abbreviation of independence of irrelevant matches, an axiom to be discussed in Section~\ref{Sec31}.} is a property one would rather not have' (p.~165). The strength of this property is clearly shown by our main theorem.

The study is structured as follows. Section~\ref{Sec2} presents the setting of the ranking problem and defines some ranking methods. In Section~\ref{Sec3}, two axioms are evoked in order to get a clear impossibility result.
Section~\ref{Sec4} investigates different ways to achieve possibility through the weakening of the axioms. Finally, some concluding remarks are given in Section~\ref{Sec5}.

\section{Preliminaries} \label{Sec2}

Consider a set of professional tennis players and their results against each other \citep{BozokiCsatoTemesi2016}. The problem is to rank them, which can be achieved by associating a score with each player. This section describes a possible mathematical model and introduces some methods.

\subsection{The ranking problem} \label{Sec21}

Let $N = \{ X_1,X_2, \dots, X_n \}$, $n \in \mathbb{N}$ be the \emph{set of objects} and $T = \left[ t_{ij} \right] \in \mathbb{R}^{n \times n}$ be a \emph{tournament matrix} such that $t_{ij} + t_{ji} \in \mathbb{N}$.
$t_{ij}$ represents the aggregated score of object $X_i$ against $X_j$, $t_{ij} / (t_{ij} + t_{ji})$ can be interpreted as the likelihood that object $X_i$ is better than object $X_j$. $t_{ii} = 0$ is assumed for all $X_i \in N$.
Possible derivations of the tournament matrix can be found in \citet{Gonzalez-DiazHendrickxLohmann2013} and \citet{Csato2015a}.

The pair $(N,T)$ is called a \emph{ranking problem}.
The set of ranking problems with $n$ objects ($|N| = n$) is denoted by $\mathcal{R}^n$.

A \emph{scoring procedure} $f$ is an $\mathcal{R}^n \to \mathbb{R}^n$ function that gives a rating $f_i(N,T)$ for each object $X_i \in N$ in any ranking problem $(N,T) \in \mathcal{R}^n$. Any scoring method immediately induces a ranking (a transitive and complete weak order on the set of $N \times N$) $\succeq$ by $f_i(N,T) \geq f_j(N,T)$ meaning that $X_i$ is ranked weakly above $X_j$, denoted by $X_i \succeq X_j$. The symmetric and asymmetric parts of $\succeq$ are denoted by $\sim$ and $\succ$, respectively: $X_i \sim X_j$ if both $X_i \succeq X_j$ and $X_i \preceq X_j$ hold, while $X_i \succ X_j$ if $X_i \succeq X_j$ holds, but $X_i \preceq X_j$ does not hold.
Every scoring method can be considered as a \emph{ranking method}. This paper discusses only ranking methods induced by scoring procedures.

A ranking problem $(N,T)$ has the skew-symmetric \emph{results matrix} $R = T - T^\top = \left[ r_{ij} \right] \in \mathbb{R}^{n \times n}$ and the symmetric \emph{matches matrix} $M = T + T^\top = \left[ m_{ij} \right] \in \mathbb{N}^{n \times n}$ such that $m_{ij}$ is the number of the comparisons between $X_i$ and $X_j$, whose outcome is given by $r_{ij}$. Matrices $R$ and $M$ also determine the tournament matrix as $T = (R + M)/2$.
In other words, a ranking problem $(N,T) \in \mathcal{R}^n$ can be denoted analogously by $(N,R,M)$ with the restriction $|r_{ij}| \leq m_{ij}$ for all $X_i,X_j \in N$, that is, the outcome of any paired comparison between two objects cannot 'exceed' their number of matches.
Although the description through results and matches matrices is not parsimonious, usually the notation $(N,R,M)$ will be used because it helps in the axiomatic approach.

The class of universal ranking problems has some meaningful subsets.
A ranking problem $(N,R,M) \in \mathcal{R}^n$ is called:
\begin{itemize}
\item
\emph{balanced} if $\sum_{X_k \in N} m_{ik} = \sum_{X_k \in N} m_{jk}$ for all $X_i,X_j \in N$. \\
The set of balanced ranking problems is denoted by $\mathcal{R}_{B}$.
\item
\emph{round-robin} if $m_{ij} = m_{k \ell}$ for all $X_i \neq X_j$ and $X_k \neq X_\ell$. \\
The set of round-robin ranking problems is denoted by $\mathcal{R}_{R}$.
\item
\emph{unweighted} if $m_{ij} \in \{ 0; 1 \}$ for all $X_i,X_j \in N$. \\
The set of unweighted ranking problems is denoted by $\mathcal{R}_{U}$.
\item
\emph{extremal} if $|r_{ij}| \in \{ 0; m_{ij} \}$ for all $X_i,X_j \in N$. \\
The set of extremal ranking problems is denoted by $\mathcal{R}_{E}$.
\end{itemize}

The first three subsets pose restrictions on the matches matrix $M$.
In a balanced ranking problem, all objects should have the same number of comparisons. A typical example is a Swiss-system tournament (provided the number of participants is even).
In a round-robin ranking problem, the number of comparisons between any pair of objects is the same. A typical example (of double round-robin) can be the qualification for soccer tournaments like UEFA European Championship \citep{Csato2018b}. It does not allow for incomplete comparisons.
Note that a round-robin ranking problem is balanced, $\mathcal{R}_{R} \subset \mathcal{R}_{B}$.
Finally, in an unweighted ranking problem, multiple comparisons are prohibited.

Extremal ranking problems restrict the results matrix $R$: the outcome of a comparison can only be a complete win ($r_{ij} = m_{ij}$), a draw  ($r_{ij} = 0$), or a maximal loss ($r_{ij} = -m_{ij}$). In other words, preferences have no intensity, however, ties are allowed.

One can also consider any intersection of these special classes.

The \emph{number of comparisons} of object $X_i \in N$ is $d_i = \sum_{X_j \in N} m_{ij}$ and the \emph{maximal number of comparisons} in the ranking problem is $m = \max_{X_i,X_j \in N} m_{ij}$. Hence:
\begin{itemize}
\item
A ranking problem is balanced if and only if $d_i = d$ for all $X_i \in N$.
\item
A ranking problem is round-robin if and only if $m_{ij} = m$ for all $X_i,X_j \in N$.
\item
A ranking problem is unweighted if and only if $m = 1$.\footnote{~While $m_{ij} \in \{ 0; 1 \}$ for all $X_i,X_j \in N$ allows for $m=0$, it leads to a meaningless ranking problem without any comparison.}
\end{itemize}

Matrix $M$ can be represented by an undirected multigraph $G := (V,E)$, where the vertex set $V$ corresponds to the object set $N$, and the number of edges between objects $X_i$ and $X_j$ is equal to $m_{ij}$, so the degree of node $X_i$ is $d_i$.
Graph $G$ is said to be the \emph{comparison multigraph} of the ranking problem $(N,R,M)$, and is independent of the results matrix $R$. The \emph{Laplacian matrix} $L = \left[ \ell_{ij} \right] \in \mathbb{R}^{n \times n}$ of graph $G$ is given by $\ell_{ij} = -m_{ij}$ for all $X_i \neq X_j$ and $\ell_{ii} = d_i$ for all $X_i \in N$.

A ranking problem $(N,R,M) \in \mathcal{R}^n$ is called \emph{connected} or \emph{unconnected} if its comparison multigraph is connected or unconnected, respectively.

\subsection{Some ranking methods} \label{Sec22}

In the following, some scoring procedures are presented. They will be used only for ranking purposes, so they can be called ranking methods.

Let $\mathbf{e} \in \mathbb{R}^n$ denote the column vector with $e_i = 1$ for all $i = 1,2, \dots ,n$.
Let $I \in \mathbb{R}^{n \times n}$ be the identity matrix.

The first scoring method does not take the comparison structure into account, it simply sums the results from the results matrix $R$.

\begin{definition} \label{Def21}
\emph{Row sum}: $\mathbf{s}(N,R,M) = R \mathbf{e}$.
\end{definition}

The following \emph{parametric} procedure has been constructed axiomatically by \citet{Chebotarev1989_eng} as an extension of the row sum method to the case of paired comparisons with missing values, and has been thoroughly analysed in \citet{Chebotarev1994}.

\begin{definition} \label{Def22}
\emph{Generalized row sum}: it is the unique solution $\mathbf{x}(\varepsilon)(N,R,M)$ of the system of linear equations $(I+ \varepsilon L) \mathbf{x}(\varepsilon)(N,R,M) = (1 + \varepsilon m n) \mathbf{s}(N,R,M)$, where $\varepsilon > 0$ is a parameter. 
\end{definition}

Generalized row sum adjusts the row sum $s_i$ by accounting for the performance of objects compared with $X_i$, and adds an infinite depth to the correction as the row sums of all objects available on a path from $X_i$ appear in the calculation. $\varepsilon$ indicates the importance attributed to this modification.
Note that generalized row sum results in row sum if $\varepsilon \to 0$: $\lim_{\varepsilon \to 0} \mathbf{x}(\varepsilon)(N,R,M) = \mathbf{s}(N,R,M)$.

The row sum and generalized row sum rankings are unique and easily computable from a system of linear equations for all ranking problems $(N,R,M) \in \mathcal{R}^n$.


The least squares method was suggested by \citet{Thurstone1927} and \citet{Horst1932}.
It is known as logarithmic least squares method in the case of incomplete multiplicative pairwise comparison matrices \citep{BozokiFulopRonyai2010}.

\begin{definition} \label{Def23}
\emph{Least squares}: it is the solution $\mathbf{q}(N,R,M)$ of the system of linear equations $L \mathbf{q}(N,R,M) = \mathbf{s}(N,R,M)$ and $\mathbf{e}^\top \mathbf{q}(N,R,M) = 0$.
\end{definition}

Generalized row sum ranking coincides with least squares ranking if $\varepsilon \to \infty$ because $\lim_{\varepsilon \to \infty} \mathbf{x}(\varepsilon)(N,R,M) = mn \mathbf{q}(N,R,M)$.

The least squares ranking is unique if and only if the ranking problem $(N,R,M) \in \mathcal{R}^n$ is connected \citep{KaiserSerlin1978, ChebotarevShamis1999, BozokiFulopRonyai2010}.
The ranking of unconnected objects may be controversial. Nonetheless, the least squares ranking can be made unique if Definition~\ref{Def23} is applied to all ranking subproblems with a connected comparison multigraph.

An extensive analysis and a graph interpretation of the least squares method, as well as further references, can be found in \citet{Csato2015a}.

\section{The impossibility result} \label{Sec3}

In this section, a natural axiom of independence and a kind of monotonicity property is recalled. 
Our main result illustrates the impossibility of satisfying the two requirements simultaneously.

\subsection{Independence of irrelevant matches} \label{Sec31}

This property appears as \emph{independence} in \citet[Axiom~III]{Rubinstein1980} and \citet[Axiom~5]{NitzanRubinstein1981} in the case of round-robin ranking problems. The name independence of irrelevant matches has been used by \citet{Gonzalez-DiazHendrickxLohmann2013}.
It deals with the effects of certain changes in the tournament matrix.

\begin{axiom} \label{Axiom31}
\emph{Independence of irrelevant matches} ($IIM$):
Let $(N,T),(N,T') \in \mathcal{R}^n$ be two ranking problems and $X_i,X_j,X_k, X_\ell \in N$ be four different objects such that $(N,T)$ and $(N,T')$ are identical but $t'_{k \ell} \neq t_{k \ell}$.
Scoring procedure $f: \mathcal{R}^n \to \mathbb{R}^n$ is called \emph{independent of irrelevant matches} if $f_i(N,T) \geq f_j(N,T) \Rightarrow f_i(N,T') \geq f_j(N,T')$.
\end{axiom}

$IIM$ means that 'remote' comparisons -- not involving objects $X_i$ and $X_j$ -- do not affect the order of $X_i$ and $X_j$.
Changing the matches matrix may lead to an unconnected ranking problem.
Property $IIM$ has a meaning if $n \geq 4$.

Sequential application of independence of irrelevant matches can lead to any ranking problem $(N,\bar{T}) \in \mathcal{R}^n$, for which $\bar{t}_{gh} = t_{gh}$ if $\{ X_g,X_h \} \cap \{ X_i, X_j \} \neq \emptyset$, but all other paired comparisons are arbitrary.

\begin{lemma} \label{Lemma31}
The row sum method is independent of irrelevant matches.
\end{lemma}

\begin{proof}
It follows from Definition~\ref{Def21}.
\end{proof}
 
\subsection{Self-consistency} \label{Sec32}

The next axiom, introduced by \citet{ChebotarevShamis1997a}, may require an extensive explanation.
It is motivated by an example using the language of preference aggregation.

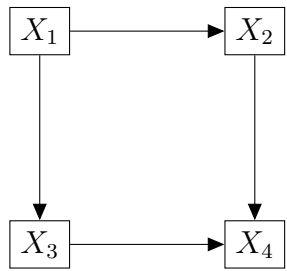
\begin{figure}[htbp]
\centering
\caption{The ranking problem of Example~\ref{Examp31}}
\label{Fig31}
\begin{tikzpicture}[scale=1, auto=center, transform shape, >=triangle 45]
\tikzstyle{every node}=[draw,shape=rectangle];
  \node (n1) at (135:2) {$X_1$};
  \node (n2) at (45:2)  {$X_2$};
  \node (n3) at (225:2) {$X_3$};
  \node (n4) at (315:2) {$X_4$};

  \foreach \from/\to in {n1/n2,n1/n3,n2/n4,n3/n4}
    \draw [->] (\from) -- (\to);
\end{tikzpicture}
\end{figure}

\begin{example} \label{Examp31}
Consider the ranking problem $(N,R,M) \in \mathcal{R}_B^4 \cap \mathcal{R}_U^4 \cap \mathcal{R}_E^4$ with results and matches matrices
\[
R = \left[
\begin{array}{cccc}
    0     & 1     & 1     & 0 \\
    -1    & 0     & 0     & 1 \\
    -1    & 0     & 0     & 1 \\
    0     & -1    & -1    & 0 \\
\end{array}
\right] \text{ and }
M = \left[
\begin{array}{cccc}
    0     & 1     & 1     & 0 \\
    1     & 0     & 0     & 1 \\
    1     & 0     & 0     & 1 \\
    0     & 1     & 1     & 0 \\
\end{array}
\right].
\]
It is shown in Figure~\ref{Fig31}: a directed edge from node $X_i$ to $X_j$ indicates a complete win of $X_i$ over $X_j$ (and a complete loss of $X_j$ against $X_i$).
This representation will be used in further examples, too.
\end{example}

The situation in Example~\ref{Examp31} can be interpreted as follows. A voter prefers alternative $X_1$ to $X_2$ and $X_3$, but says nothing about $X_4$. Another voter prefers $X_2$ to $X_3$ and $X_4$, but has no opinion on $X_1$.

Although it is difficult to make a good decision on the basis of such incomplete preferences, sometimes it cannot be avoided. It leads to the question, which principles should be followed by the final ranking of the objects. It seems reasonable that $X_i$ should be judged better than $X_j$ if one of the following holds:
\begin{enumerate}[label=\ding{64}\arabic*]
\item \label{SC_con1}
$X_i$ achieves better results against the same objects;
\item \label{SC_con2}
$X_i$ achieves better results against objects with the same strength;
\item \label{SC_con3}
$X_i$ achieves the same results against stronger objects;
\item \label{SC_con4}
$X_i$ achieves better results against stronger objects.
\end{enumerate}
Furthermore, $X_i$ should have the same rank as $X_j$ if one of the following holds:
\begin{enumerate}[resume,label=\ding{64}\arabic*]
\item \label{SC_con5}
$X_i$ achieves the same results against the same objects;
\item \label{SC_con6}
$X_i$ achieves the same results against objects with the same strength.
\end{enumerate}

In order to apply these principles, one should measure the strength of objects. It is provided by the scoring method itself, hence the name of this axiom is \emph{self-consistency}.
Consequently, condition~\ref{SC_con1} is a special case of condition~\ref{SC_con2} (the same objects have naturally the same strength) as well as condition~\ref{SC_con5} is implied by condition~\ref{SC_con6}.

What does self-consistency mean in Example~\ref{Examp31}?
First, $X_2 \sim X_3$ due to condition~\ref{SC_con5}.
Second, $X_1 \succ X_4$ should hold since condition~\ref{SC_con1} as $r_{12} > r_{42}$ and $r_{13} > r_{43}$.
The requirements above can also be applied to objects which have different opponents.
Assume that $X_1 \preceq X_2$. Then condition~\ref{SC_con4} results in $X_1 \succ X_2$ because of $X_2 \succeq X_1$, $r_{12} > r_{21}$ and $X_3 \sim X_2 \succeq X_1 \succ X_4$, $r_{13} = r_{24}$. It is a contradiction, therefore $X_1 \succ (X_2 \sim X_3)$.
Similarly, assume that $X_2 \preceq X_4$. Then condition~\ref{SC_con4} results in $X_2 \succ X_4$ because of $X_1 \succ X_3$ (derived above), $r_{21} = r_{43}$ and $X_4 \succeq X_2 \sim X_3$, $r_{24} > r_{43}$. It is a contradiction, therefore $(X_2 \sim X_3) \succ X_4$.
To summarize, only the ranking $X_1 \succ (X_2 \sim X_3) \succ X_4$ is allowed by self-consistency.

The above requirement can be formalized in the following way.

\begin{definition} \label{Def31}
\emph{Opponent set}:
Let $(N,R,M) \in \mathcal{R}_U^n$ be an unweighted ranking problem. The \emph{opponent set} of object $X_i$ is $O_i = \{ X_j: m_{ij} = 1 \}$
\end{definition}

Objects of the opponent set $O_i$ are called the \emph{opponents} of $X_i$.
Note that $|O_i| = |O_j|$ for all $X_i, X_j \in N$ if and only if the ranking problem is balanced.

\begin{notation} \label{Not31}
Consider an unweighted ranking problem $(N,R,M) \in \mathcal{R}_U^n$ such that $X_i, X_j \in N$ are two different objects and $g: O_i \leftrightarrow O_j$ is a one-to-one correspondence between the opponents of $X_i$ and $X_j$, consequently, $|O_i| = |O_j|$.
Then $\mathfrak{g} : \{k: X_k \in O_i \} \leftrightarrow \{\ell: X_\ell \in O_j \}$ is given by $X_{\mathfrak{g}(k)} = g(X_k)$.
\end{notation}

In order to make judgements like an object has stronger opponents, at least a partial order among opponent sets should be introduced. 

\begin{definition} \label{Def32}
\emph{Partial order of opponent sets}:
Let $(N,R,M) \in \mathcal{R}^n$ be a ranking problem and $f: \mathcal{R}^n \to \mathbb{R}^n$ be a scoring procedure.
Opponents of $X_i$ are at least as strong as opponents of $X_j$, denoted by $O_i \succeq O_j$, if there exists a one-to-one correspondence $g:O_i \leftrightarrow O_j$ such that $f_k(N,R,M) \geq f_{\mathfrak{g}(k)}(N,R,M)$ for all $X_k \in O_i$.
\end{definition}

For instance, $O_1 \sim O_4$ and $O_2 \sim O_3$ in Example~\ref{Examp31}, whereas $O_1$ and $O_2$ are not comparable.

Therefore, conditions~\ref{SC_con1}-\ref{SC_con6} never imply $X_i \succeq X_j$ if $O_i \prec O_j$ since an object with a weaker opponent set cannot be judged better.

Opponent sets have been defined only in the case of unweighted ranking problems, but self-consistency can be applied to objects which have the same number of comparisons, too. The extension is achieved by a decomposition of ranking problems.

\begin{definition} \label{Def33}
\emph{Sum of ranking problems}:
Let $(N,R,M),(N,R',M') \in \mathcal{R}^n$ be two ranking problems with the same object set $N$. The \emph{sum} of these ranking problems is the ranking problem $(N,R+R',M+M') \in \mathcal{R}^n$.
\end{definition}

Summing of ranking problems may have a natural interpretation. For example, they can contain the preferences of voters in two cities of the same country or the paired comparisons of players in the first and second half of the season.

Definition~\ref{Def33} means that any ranking problem can be decomposed into unweighted ranking problems, in other words, it can be obtained as a sum of unweighted ranking problems.
However, while the sum of ranking problems is unique, a ranking problem may have a number of possible decompositions.

\begin{notation} \label{Not32}
Let $(N,R^{(p)},M^{(p)}) \in \mathcal{R}_U^n$ be an unweighted ranking problem.
The opponent set of object $X_i$ is $O_i^{(p)}$.
Let $X_i, X_j \in N$ be two different objects and $g^{(p)}: O_i^{(p)} \leftrightarrow O_j^{(p)}$ be a one-to-one correspondence between the opponents of $X_i$ and $X_j$.
Then $\mathfrak{g}^{(p)}: \{k: X_k \in O_i^{(p)} \} \leftrightarrow \{\ell: X_\ell \in O_j^{(p)} \}$ is given by $X_{\mathfrak{g}^{(p)}(k)} = g^{(p)}(X_k)$.
\end{notation}

\begin{axiom} \label{Axiom32}
\emph{Self-consistency} ($SC$) \citep{ChebotarevShamis1997a}:
A scoring procedure $f: \mathcal{R}^n \to \mathbb{R}^n$ is called \emph{self-consistent} if the following implication holds for any ranking problem $(N,R,M) \in \mathcal{R}^n$ and for any objects $X_i,X_j \in N$:
if there exists a decomposition of the ranking problem $(N,R,M)$ into $m$ unweighted ranking problems -- that is, $R = \sum_{p=1}^m R^{(p)}$, $M = \sum_{p=1}^m M^{(p)}$, and $(N,R^{(p)},M^{(p)}) \in \mathcal{R}_U^n$ is an unweighted ranking problem for all $p = 1,2, \dots ,m$ -- in a way that enables a one-to-one mapping $g^{(p)}$ from $O^{(p)}_i$ onto $O^{(p)}_j$ such that $r_{ik}^{(p)} \geq r_{j \mathfrak{g}^{(p)}(k)}^{(p)}$ and $f_k(N,R,M) \geq f_{\mathfrak{g}^{(p)}(k)}(N,R,M)$ for all $p = 1,2, \dots ,m$ and $X_k \in O_i^{(p)}$, then
$f_i(N,R,M) \geq f_{j}(N,R,M)$, furthermore, $f_i(N,R,M) > f_{j}(N,R,M)$ if $r_{ik}^{(p)} > r_{j \mathfrak{g}^{(p)}(k)}^{(p)}$ or $f_k(N,R,M) > f_{\mathfrak{g}^{(p)}(k)}(N,R,M)$ for at least one $1 \leq p \leq m$ and $X_k \in O_i^{(p)}$.
\end{axiom}

Self-consistency formalizes conditions~\ref{SC_con1}-\ref{SC_con6}: if object $X_i$ is obviously not worse than object $X_j$, then it is not ranked lower, furthermore, if it is better, then it is ranked higher.
Self-consistency can also be interpreted as a property of a ranking.

The application of self-consistency is nontrivial because of the various opportunities for decomposition into unweighted ranking problems.
However, it may restrict the relative ranking of objects $X_i$ and $X_j$ only if $d_i = d_j$ since there should exist a one-to-one mapping between $O_i^{(p)}$ and $O_j^{(p)}$ for all $p = 1,2, \dots ,m$.
Thus $SC$ does not fully determine a ranking, even on the set of balanced ranking problems.

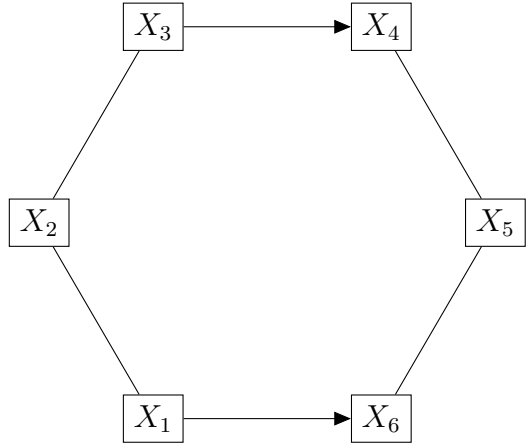
\begin{figure}[htbp]
\centering
\caption{The ranking problem of Example~\ref{Examp32}}
\label{Fig32}
\begin{tikzpicture}[scale=1, auto=center, transform shape, >=triangle 45]
\tikzstyle{every node}=[draw,shape=rectangle];
  \node (n1) at (240:3) {$X_1$};
  \node (n2) at (180:3) {$X_2$};
  \node (n3) at (120:3) {$X_3$};
  \node (n4) at (60:3)  {$X_4$};
  \node (n5) at (0:3)   {$X_5$};
  \node (n6) at (300:3) {$X_6$};

  \foreach \from/\to in {n1/n2,n2/n3,n4/n5,n5/n6}
    \draw (\from) -- (\to);
  \foreach \from/\to in {n1/n6,n3/n4}
    \draw [->] (\from) -- (\to);
\end{tikzpicture}
\end{figure}

\begin{example} \label{Examp32}
Let $(N,R,M) \in \mathcal{R}_B^6 \cap \mathcal{R}_U^6 \cap \mathcal{R}_E^6$ be the ranking problem in Figure~\ref{Fig32}: a directed edge from node $X_i$ to $X_j$ indicates a complete win of $X_i$ over $X_j$ in one comparison (as in Example~\ref{Examp31}) and an undirected edge from node $X_i$ to $X_j$ represents a draw in one comparison between the two objects.
\end{example}

\begin{proposition} \label{Prop31}
Self-consistency does not fully characterize a ranking method on the set of balanced, unweighted and extremal ranking problems $\mathcal{R}_B \cap \mathcal{R}_U \cap \mathcal{R}_E$.
\end{proposition}

\begin{proof}
The statement can be verified by an example where at least two rankings are allowed by $SC$, we use Example~\ref{Examp32} for this purpose.
Consider the ranking $\succeq^1$ such that $(X_1 \sim^1 X_2 \sim^1 X_3) \succ^1 (X_4 \sim^1 X_5 \sim^1 X_6)$.
The opponent sets are $O_1 = \{ X_2, X_6 \}$, $O_2 = \{ X_1, X_3 \}$, $O_3 = \{ X_2, X_4 \}$, $O_4 = \{ X_3, X_5 \}$, $O_5 = \{ X_4, X_6 \}$ and $O_6 = \{ X_1, X_5 \}$, so $O_2 \succ (O_1 \sim O_3 \sim O_4 \sim O_6) \succ O_5$.
The results of $X_1$ and $X_3$ are $(0;1)$, the results of $X_2$ and $X_5$ are $(0;0)$, while the results of $X_4$ and $X_6$ are $(-1;0)$.
For objects with the same results, $SC$ implies $X_1 \sim X_3$, $X_4 \sim X_6$ and $X_2 \succ X_5$ (conditions~\ref{SC_con3} and \ref{SC_con6}), which hold in  $\succeq^1$.
For objects with different results, $SC$ leads to $X_2 \succ X_4$, $X_3 \succ X_4$, and $X_3 \succ X_5$ after taking the strength of opponents into account (condition~\ref{SC_con2}). These requirements are also met by the ranking $\succeq^1$.
Self-consistency imposes no other restrictions, therefore the ranking $\succeq^1$ satisfies it.

Now consider the ranking $\succeq^2$ such that $X_2 \prec^2 (X_1 \sim^2 X_3) \prec^2 (X_4 \sim^2 X_6) \prec^2 X_5$.
The opponent sets remain the same, but their partial order is given now as $O_2 \prec (O_4 \sim O_6)$, $O_2 \prec O_5$, $(O_1 \sim O_3) \prec (O_4 \sim O_6)$ and $(O_1 \sim O_3) \prec O_5$ (the opponents of $X_1$ and $X_2$, as well as $X_4$ and $X_5$, cannot be compared).
For objects with the same results, $SC$ implies $X_1 \sim X_3$, $X_4 \sim X_6$ and $X_2 \prec X_5$ (conditions~\ref{SC_con3} and \ref{SC_con6}), which hold in  $\succeq^2$.
For objects with different results, $SC$ leads to $X_1 \succ X_2$ after taking the strength of opponents into account (condition~\ref{SC_con2}). This condition is also met by the ranking $\succeq^2$.
Self-consistency imposes no other restrictions, therefore the ranking $\succeq^2$ also satisfies this axiom.

To conclude, rankings $\succeq^1$ and $\succeq^2$ are self-consistent. The ranking obtained by reversing $\succeq^2$ meets $SC$, too.
\end{proof}

\begin{lemma} \label{Lemma32}
The generalized row sum and least squares methods are self-consistent.
\end{lemma}

\begin{proof}
See \citet[Theorem~5]{ChebotarevShamis1998a}.
\end{proof}

\citet[Theorem~5]{ChebotarevShamis1998a} provide a characterization of self-consistent scoring procedures, while \citet[Table~2]{ChebotarevShamis1998a} gives some further examples.

\subsection{The connection of independence of irrelevant matches and self-consistency} \label{Sec33}

So far we have discussed two axioms, $IIM$ and $SC$. It turns out that they cannot be satisfied at the same time.

\begin{figure}[htbp]
\centering
\caption{The ranking problems of Example~\ref{Examp33}}
\label{Fig33}
  
\begin{subfigure}{.5\textwidth}
  \centering
  \subcaption{Ranking problem $(N,R,M)$}
  \label{Fig33a}
\begin{tikzpicture}[scale=1, auto=center, transform shape, >=triangle 45]
\tikzstyle{every node}=[draw,shape=rectangle]; 
  \node (n1) at (135:2) {$X_1$};
  \node (n2) at (45:2)  {$X_2$};
  \node (n3) at (315:2) {$X_3$};
  \node (n4) at (225:2) {$X_4$};

  \foreach \from/\to in {n1/n2,n1/n4,n2/n3}
    \draw (\from) -- (\to);
  \draw [->] (n4) -- (n3);
\end{tikzpicture}
\end{subfigure}
\begin{subfigure}{.5\textwidth}
  \centering
  \subcaption{Ranking problem $(N,R',M)$}
  \label{Fig33b}
\begin{tikzpicture}[scale=1, auto=center, transform shape, >=triangle 45]
\tikzstyle{every node}=[draw,shape=rectangle];
  \node (n1) at (135:2) {$X_1$};
  \node (n2) at (45:2)  {$X_2$};
  \node (n3) at (315:2) {$X_3$};
  \node (n4) at (225:2) {$X_4$};

  \foreach \from/\to in {n1/n2,n1/n4,n2/n3}
    \draw (\from) -- (\to);
  \draw [->] (n3) -- (n4);
\end{tikzpicture}
\end{subfigure}
\end{figure}
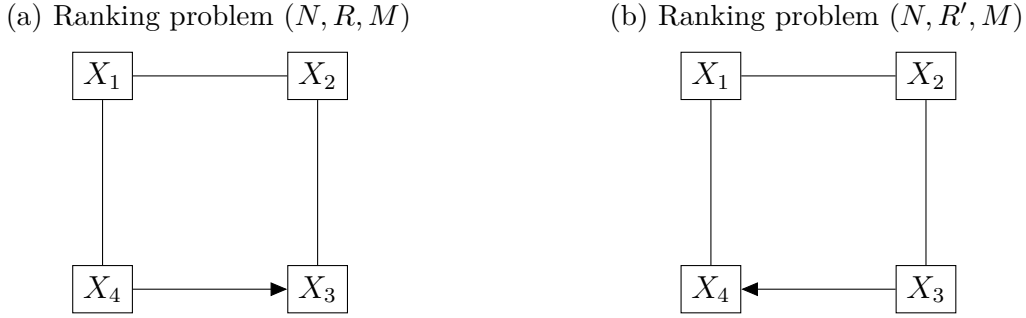

\begin{example} \label{Examp33}
Let $(N,R,M), (N,R',M) \in \mathcal{R}_B^4 \cap \mathcal{R}_U^4 \cap \mathcal{R}_E^4$ be the ranking problems in Figure~\ref{Fig33} with the results and matches matrices
\[
R = \left[
\begin{array}{cccc}
    0     & 0     & 0     & 0 \\
    0     & 0     & 0     & 0 \\
    0     & 0     & 0     & 1 \\
    0     & 0     & -1    & 0 \\
\end{array}
\right], \,
R' = \left[
\begin{array}{cccc}
    0     & 0     & 0     & 0 \\
    0     & 0     & 0     & 0 \\
    0     & 0     & 0     & -1 \\
    0     & 0     & 1     & 0 \\
\end{array}
\right], \text{ and }
M = \left[
\begin{array}{cccc}
    0     & 1     & 0     & 1 \\
    1     & 0     & 1     & 0 \\
    0     & 1     & 0     & 1 \\
    1     & 0     & 1     & 0 \\
\end{array}
\right].
\]
\end{example}

\begin{theorem} \label{Theo31}
There exists no scoring procedure that is independent of irrelevant matches and self-consistent.
\end{theorem}

\begin{proof}
The contradiction of the two properties is proved by Example~\ref{Examp33}.
The opponent sets are $O_1 = O_3 = \{ X_2, X_4 \}$ and $O_2 = O_4 = \{ X_1, X_3 \}$ in both ranking problems.
Assume to the contrary that there exists a scoring procedure $f: \mathcal{R}^n \to \mathbb{R}^n$, which is independent of irrelevant matches and self-consistent.
$IIM$ means that $f_1(N,R,M) \geq f_2(N,R,M) \iff f_1(N,R',M) \geq f_2(N,R',M)$.

\begin{enumerate}[label=\emph{\alph*})]
\item \label{Enum_a}
Consider the (identity) one-to-one mapping $g_{13}: O_1 \leftrightarrow O_3$, where $g_{13}(X_2) = X_2$ and $g_{13}(X_4) = X_4$. Since $r_{12} = r_{42} = 0$ and $0 = r_{14} > r_{34} = -1$, $g_{13}$ satisfies condition~\ref{SC_con1} of $SC$, hence $f_1(N,R,M) > f_3(N,R,M)$.

\item \label{Enum_b}
Consider the (identity) one-to-one mapping $g_{42}: O_4 \leftrightarrow O_2$, where $g_{42}(X_1) = X_1$ and $g_{42}(X_3) = X_3$. Since $r_{41} = r_{21} = 0$ and $1 = r_{43} > r_{23} = 0$, $g_{42}$ satisfies condition~\ref{SC_con1} of $SC$, hence $f_4(N,R,M) > f_2(N,R,M)$.

\item \label{Enum_c}
Suppose that $f_2(N,R,M) \geq f_1(N,R,M)$, implying $f_4(N,R,M) > f_3(N,R,M)$.
Consider the one-to-one correspondence $g_{12}: O_1 \leftrightarrow O_2$, where $g_{12}(X_2) = X_1$ and $g_{12}(X_4) = X_3$. Since $r_{12} = r_{21} = 0$ and $r_{14} = r_{23} = 0$, $g_{12}$ satisfies condition~\ref{SC_con3} of $SC$, hence $f_1(N,R,M) > f_2(N,R,M)$. It is a contradiction.
\end{enumerate}

Thus only $f_1(N,R,M) > f_2(N,R,M)$ is allowed.

Note that ranking problem $(N,R',M)$ can be obtained from $(N,R,M)$ by the permutation $\sigma: N \to N$ such that $\sigma(X_1) = X_2$, $\sigma(X_2) = X_1$, $\sigma(X_3) = X_4$ and $\sigma(X_4) = X_3$. The above argument results in $f_2(N,R',M) > f_1(N,R',M)$, contrary to independence of irrelevant matches.

To conclude, no scoring procedure can meet $IIM$ and $SC$ simultaneously.
\end{proof}

\begin{corollary} \label{Col31}
The row sum method violates self-consistency.
\end{corollary}

\begin{proof}
It is an immediate consequence of Lemma~\ref{Lemma31} and Theorem~\ref{Theo31}.
\end{proof}

\begin{corollary} \label{Col32}
The generalized row sum and least squares methods violate independence of irrelevant matches.
\end{corollary}

\begin{proof}
It follows from Lemma~\ref{Lemma32} and Theorem~\ref{Theo31}.
\end{proof}

A set of axioms is said to be \emph{logically independent} if none of them are implied by the others.

\begin{corollary} \label{Col33}
$IIM$ and $SC$ are logically independent axioms.
\end{corollary}

\begin{proof}
It is a consequence of Corollaries~\ref{Col31} and \ref{Col32}.
\end{proof}

\section{How to achieve possibility?} \label{Sec4}

Impossibility results, like the one in Theorem~\ref{Theo31}, can be avoided in at least two ways: by introducing some restrictions on the class of ranking problems considered, or by weakening of one or more axioms.

\subsection{Domain restrictions} \label{Sec41}

Besides the natural subclasses of ranking problems introduced in Section~\ref{Sec21}, the number of objects can be limited, too.

\begin{proposition} \label{Prop41}
The generalized row sum and least squares methods are independent of irrelevant matches and self-consistent on the set of ranking problems with at most three objects $\mathcal{R}^n | n \leq 3$.
\end{proposition}

\begin{proof}
$IIM$ has no meaning on the set $\mathcal{R}^n | n \leq 3$, so any self-consistent scoring procedure is appropriate, thus Lemma~\ref{Lemma32} provides the result.
\end{proof}

Proposition~\ref{Prop41} has some significance since ranking is not trivial if $n=3$.
However, if at least four objects are allowed, the situation is much more severe.

\begin{proposition} \label{Prop42}
There exists no scoring procedure that is independent of irrelevant matches and self-consistent on the set of balanced, unweighted and extremal ranking problems with four objects $\mathcal{R}_B^4 \cap \mathcal{R}_U^4 \cap \mathcal{R}_E^4$.
\end{proposition}

\begin{proof}
The ranking problems of Example~\ref{Examp33}, used for verifying the impossibility in Theorem~\ref{Theo31}, are from the set $\mathcal{R}_B^4 \cap \mathcal{R}_U^4 \cap \mathcal{R}_E^4$.
\end{proof}

Proposition~\ref{Prop42} does not deal with the class of round-robin ranking problems. Then another possibility result emerges.

\begin{proposition} \label{Prop43}
The row sum, generalized row sum and least squares methods are independent of irrelevant matches and self-consistent on the set of round-robin ranking problems $\mathcal{R}_R$.
\end{proposition}

\begin{proof}
Due to axioms \emph{agreement} \citep[Property~3]{Chebotarev1994} and \emph{score consistency} \citep{Gonzalez-DiazHendrickxLohmann2013}, the generalized row sum and least squares ranking methods coincide with the row sum on the set of $\mathcal{R}_R$, so Lemmata~\ref{Lemma31} and \ref{Lemma32} provide $IIM$ and $SC$, respectively.
\end{proof}

Perhaps it is not by chance that characterizations of the row sum method were suggested on this -- or even more restricted -- domain \citep{Young1974, HanssonSahlquist1976, Rubinstein1980, NitzanRubinstein1981, Henriet1985, Bouyssou1992}.

\subsection{Weakening of independence of irrelevant matches} \label{Sec42}

For the relaxation of $IIM$, a property discussed by \citet{Chebotarev1994} will be used.

\begin{definition} \label{Def41}
\emph{Macrovertex} \citep[Definition~3.1]{Chebotarev1994}:
Let $(N,R,M) \in \mathcal{R}^n$ be a ranking problem.
Object set $V \subseteq N$ is called \emph{macrovertex} if $m_{ik} = m_{jk}$ for all $X_i, X_j \in V$ and $X_k \in N \setminus V$.
\end{definition}

Objects in a macrovertex have the same number of comparisons against any object outside the macrovertex. The comparison structure in $V$ and $N \setminus V$ can be arbitrary. The existence of a macrovertex depends only on the matches matrix $M$, or, in other words, on the comparison multigraph of the ranking problem.

\begin{axiom} \label{Axiom41}
\emph{Macrovertex independence ($MVI$)} \citep[Property~8]{Chebotarev1994}:
Let $V \subseteq N$ be a macrovertex in ranking problems $(N,T),(N,T') \in \mathcal{R}^n$ and $X_i, X_j \in V$ be two different objects such that $(N,T)$ and $(N,T')$ are identical but $t'_{ij} \neq t_{ij}$.
Scoring procedure $f: \mathcal{R}^n \to \mathbb{R}^n$ is called \emph{macrovertex independent} if $f_k(N,T) \geq f_\ell(N,T) \Rightarrow f_k(N,T') \geq f_\ell(N,T')$ for all $X_k, X_\ell \in N \setminus V$.
\end{axiom}

Macrovertex independence says that the order of objects outside a macrovertex is independent of the number and result of comparisons between the objects inside the macrovertex.

\begin{corollary} \label{Col41}
$IIM$ implies $MVI$.
\end{corollary}

Note that if $V$ is a macrovertex, then $N \setminus V$ is not necessarily another macrovertex. Hence the 'dual' of property $MVI$ can be introduced.

\begin{axiom} \label{Axiom42}
\emph{Macrovertex autonomy ($MVA$)}:
Let $V \subseteq N$ be a macrovertex in ranking problems $(N,T),(N,T') \in \mathcal{R}^n$ and $X_k, X_\ell \in N \setminus V$ be two different objects such that $(N,T)$ and $(N,T')$ are identical but $t'_{k \ell} \neq t_{k \ell}$.
Scoring procedure $f: \mathcal{R}^n \to \mathbb{R}^n$ is called \emph{macrovertex autonomous} if $f_i(N,T) \geq f_j(N,T) \Rightarrow f_i(N,T') \geq f_j(N,T')$ for all $X_i, X_j \in V$.
\end{axiom}

Macrovertex autonomy says that the order of objects inside a macrovertex is not influenced by the number and result of comparisons between the objects outside the macrovertex.

\begin{corollary} \label{Col42}
$IIM$ implies $MVA$.
\end{corollary}

Similarly to $IIM$, changing the matches matrix -- as allowed by properties $MVI$ and $MVA$ -- may lead to an unconnected ranking problem.

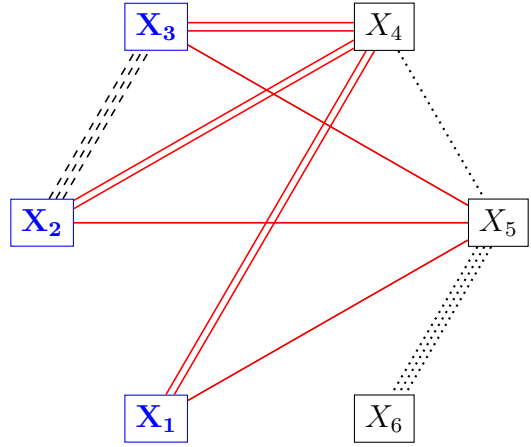
\begin{figure}[htbp]
\centering
\caption{The comparison multigraph of Example~\ref{Examp41}}
\label{Fig41}

\begin{tikzpicture}[scale=1, auto=center, transform shape]
\tikzstyle{every node}=[draw,shape=rectangle];
  \node[color=blue] (n1) at (240:3) {\textcolor{blue}{$\mathbf{X_1}$}};
  \node[color=blue] (n2) at (180:3) {\textcolor{blue}{$\mathbf{X_2}$}};
  \node[color=blue] (n3) at (120:3) {\textcolor{blue}{$\mathbf{X_3}$}};
  \node (n4) at (60:3)  {$X_4$};
  \node (n5) at (0:3) {$X_5$};
  \node (n6) at (300:3)   {$X_6$};
  
  \foreach \from/\to in {n1/n5,n2/n5,n3/n5}
    \draw (\from) -- (\to);
\draw[transform canvas={xshift=0.3ex},color=red,semithick](n1) -- (n4);
\draw[transform canvas={xshift=-0.3ex},color=red,semithick](n1) -- (n4);
\draw[color=red,semithick](n1) -- (n5);
\draw[transform canvas={yshift=0.3ex},color=red,semithick](n2) -- (n4);
\draw[transform canvas={yshift=-0.3ex},color=red,semithick](n2) -- (n4);
\draw[color=red,semithick](n2) -- (n5);
\draw[transform canvas={yshift=0.3ex},color=red,semithick](n3) -- (n4);
\draw[transform canvas={yshift=-0.3ex},color=red,semithick](n3) -- (n4);
\draw[color=red,semithick](n3) -- (n5);
\draw[transform canvas={xshift=0.5ex},dashed,semithick](n2) -- (n3);
\draw[transform canvas={xshift=0ex},dashed,semithick](n2) -- (n3);
\draw[transform canvas={xshift=-0.5ex},dashed,semithick](n2) -- (n3);
\draw[transform canvas={xshift=0.5ex},dotted,thick](n5) -- (n6);
\draw[transform canvas={xshift=0ex},dotted,thick](n5) -- (n6);
\draw[transform canvas={xshift=-0.5ex},dotted,thick](n5) -- (n6);
\draw[transform canvas={yshift=0ex},dotted,thick](n4) -- (n5);
\end{tikzpicture}
\end{figure}

\begin{example} \label{Examp41}
Consider a ranking problem with the comparison multigraph in Figure~\ref{Fig41}.
The object set $V = \{ \mathbf{X_1,X_2,X_3} \}$ is a macrovertex as the number of (red) edges from any node inside $V$ to any node outside $V$ is the same (two to $X_4$, one to $X_5$, and zero to $X_6$). $V$ remains a macrovertex if comparisons inside $V$ (represented by dashed edges) or comparisons outside $V$ (dotted edges) are changed.

Macrovertex independence requires that the relative ranking of $X_4$, $X_5$, and $X_6$ does not depend on the number and result of comparisons between the objects $X_1$, $X_2$, and $X_3$.

Macrovertex autonomy requires that the relative ranking of $X_1$, $X_2$, and $X_3$ does not depend on the number and result of comparisons between the objects $X_4$, $X_5$, and $X_6$ 

The implications of $MVI$ and $MVA$ are clearly different since object set $N \setminus V = \{ X_4, X_5, X_6 \}$ is not a macrovertex because $m_{14} = 2 \neq 1 = m_{15}$.
\end{example}

\begin{corollary} \label{Col43}
The row sum method satisfies macrovertex independence and macrovertex autonomy.
\end{corollary}

\begin{proof}
It is an immediate consequence of Lemma~\ref{Lemma31} and Corollaries~\ref{Col41} and \ref{Col42}.
\end{proof}

\begin{lemma} \label{Lemma41}
The generalized row sum and least squares methods are macrovertex independent and macrovertex autonomous.
\end{lemma}

\begin{proof}
\citet[Property~8]{Chebotarev1994} has shown that generalized row sum satisfies $MVI$. The proof remains valid in the limit $\varepsilon \to \infty$ if the least squares ranking is defined to be unique, for instance, the sum of ratings of objects in all components of the comparison multigraph is zero.

Consider $MVA$.
Let $\mathbf{s} = \mathbf{s}(N,T)$, $\mathbf{s}' = \mathbf{s}(N,T')$, $\mathbf{x} = \mathbf{x}(\varepsilon)(N,T)$, $\mathbf{x}' = \mathbf{x}(\varepsilon)(N,T')$ and $\mathbf{q} = \mathbf{q}(N,T)$, $\mathbf{q}' = \mathbf{q}(N,T')$.
Let $V$ be a macrovertex and $X_i, X_j \in V$ be two arbitrary objects.
Suppose to the contrary that $x_i \geq x_j$, but $x_i' < x_j'$, hence $x_i' - x_i <  x_j' - x_j$.
Let $x_k' - x_k = \max_{X_g \in V} (x_g' - x_g)$ and $x_\ell' - x_\ell = \min_{X_g \in V} (x_g' - x_g)$, therefore $x_k' - x_k > x_\ell' - x_\ell$ and $x_k' - x_k \geq x_g' - x_g \geq x_\ell' - x_\ell$ for any object $X_g \in V$.

For object $X_k$, definition~\ref{Def22} results  in
\begin{equation} \label{eq1}
x_k = (1+\varepsilon m n)s_k + \varepsilon \sum_{X_g \in V} m_{kg} (x_g - x_k) + \varepsilon \sum_{X_h \in N \setminus V} m_{kh} (x_h - x_k).
\end{equation}
Apply \eqref{eq1} for object $X_\ell$. The difference of these two equations is
\begin{eqnarray} \label{eq2}
x_k - x_\ell & = & (1+\varepsilon m n) (s_k - s_\ell) + \varepsilon \sum_{X_g \in V} \left[ m_{kg} (x_g - x_k) - m_{\ell g} (x_g - x_\ell) \right] + \nonumber \\
& & + \varepsilon \sum_{X_h \in N \setminus V} \left[ m_{kh} (x_h - x_k) -  m_{\ell h} (x_h - x_\ell) \right].
\end{eqnarray}
Note that $m_{kh} = m_{\ell h}$ for all $X_h \in N \setminus V$ since $V$ is a macrovertex, therefore \eqref{eq2} is equivalent to
\begin{eqnarray} \label{eq3}
\left( 1 + \varepsilon \sum_{X_h \in N \setminus V} m_{kh} \right) \left( x_k - x_\ell \right) & = & (1+\varepsilon m n) (s_k - s_\ell) + \nonumber \\
 & & + \varepsilon \sum_{X_g \in V} \left[ m_{kg} (x_g - x_k) - m_{\ell g} (x_g - x_\ell) \right].
\end{eqnarray}
Apply \eqref{eq3} for the ranking problem $(N,T')$:
\begin{eqnarray} \label{eq4}
\left( 1 + \varepsilon \sum_{X_h \in N \setminus V} m_{kh}' \right) \left( x_k' - x_\ell' \right) & = & (1+\varepsilon m n) (s_k' - s_\ell') + \nonumber \\
 & & + \varepsilon \sum_{X_g \in V} \left[ m_{kg}' (x_g' - x_k') - m_{\ell g}' (x_g' - x_\ell') \right].
\end{eqnarray}
Let $\Delta_{ij} = (x_i' - x_j') - (x_i - x_j)$ for all $X_i, X_j \in V$. Note that $m_{kh}' = m_{kh}$ for all $X_h \in N \setminus V$, $m_{kg}' = m_{kg}$ and $m_{\ell g}' = m_{\ell g}$ for all $X_g \in V$ as well as $s_k' = s_k$ and $s_\ell' = s_\ell$ since only comparisons outside $V$ may change. Take the difference of \eqref{eq4} and \eqref{eq3}
\begin{equation} \label{eq5}
\left( 1 + \varepsilon \sum_{X_h \in N \setminus V} m_{kh} \right) \Delta_{k \ell} = \varepsilon \sum_{X_g \in V} \left( m_{kg} \Delta_{gk} - m_{\ell g} \Delta_{g \ell} \right).
\end{equation}
Due to the choice of indices $k$ and $\ell$, $\Delta_{k \ell} > 0$ and $\Delta_{gk} \leq 0$, $\Delta_{g \ell} \geq 0$. It means that the left-hand side of \eqref{eq5} is positive, while its right-hand side is nonpositive, leading to a contradiction. Therefore only $x_i' - x_i =  x_j' - x_j$, the condition required by $MVA$, can hold.

The same derivation can be implemented for the least squares method. With the notation $\Delta_{ij} = (q_i' - q_j') - (q_i - q_j)$ for all $X_i, X_j \in V$, we get -- analogously to \eqref{eq5} as $\varepsilon \to \infty$ --
\begin{equation} \label{eq6}
\sum_{X_h \in N \setminus V} m_{kh} \Delta_{k \ell} = \sum_{X_g \in V} \left( m_{kg} \Delta_{gk} - m_{\ell g} \Delta_{g \ell} \right).
\end{equation}
But $\Delta_{k \ell} > 0$, $\Delta_{gk} \leq 0$, and $\Delta_{g \ell} \geq 0$ is not enough for a contradiction now: \eqref{eq6} may hold if $\sum_{X_h \in N \setminus V} m_{kh} = 0$, namely, $X_k$ is not connected to any object outside the macrovertex $V$ as well as $\Delta_{gk} = 0$ and $\Delta_{g \ell} = 0$ when $m_{kg} = m_{\ell g} > 0$.
However, if there exists no object $X_g \in N \setminus V$ such that $m_{kg} = m_{\ell g} > 0$, then there is no connection between object sets $V$ and $N \setminus V$ since $V$ is a macrovertex, and we have two independent ranking subproblems, where the least squares ranking is unique according to the extension of definition~\ref{Def23}, so $MVA$ holds.
On the other hand, if there exists an object $X_g \in N \setminus V$ such that $m_{kg} = m_{\ell g} > 0$, then $\Delta_{gk} = 0$ and $\Delta_{g \ell} = 0$, but $\Delta_{k \ell} = \Delta_{g \ell} - \Delta_{gk} > 0$, which is a contradiction.
Therefore $q_i' - q_i =  q_j' - q_j$, the condition required by $MVA$, holds.
\end{proof}

Lemma~\ref{Lemma41} leads to another possibility result.

\begin{proposition} \label{Prop44}
The generalized row sum and least squares methods are macrovertex autonomous, macrovertex independent and self-consistent.
\end{proposition}

This statement turns out to be more general than the one obtained by restricting the domain to round-robin ranking problems in Proposition~\ref{Prop43}.

\begin{corollary} \label{Col44}
$MVA$ or $MVI$ implies $IIM$ on the domain of round-robin ranking problems $\mathcal{R}_R$.
\end{corollary}

\begin{proof}
Let $(N,T),(N,T') \in \mathcal{R}_R^n$ be two ranking problems and $X_i,X_j,X_k, X_\ell \in N$ be four different objects such that $(N,T)$ and $(N,T')$ are identical but $t'_{k \ell} \neq t_{k \ell}$.

Consider the macrovertex $V = \{ X_i,X_j \}$. Macrovertex autonomy means $f_i(N,T) \geq f_j(N,T) \Rightarrow f_i(N,T') \geq f_j(N,T')$, the condition required by $IIM$.

Consider the macrovertex $V' = \{ X_k,X_\ell \}$. Macrovertex independence means $f_i(N,T) \geq f_j(N,T) \Rightarrow f_i(N,T') \geq f_j(N,T')$, the condition required by $IIM$.
\end{proof}

\subsection{Weakening of self-consistency} \label{Sec43}

We think self-consistency is more difficult to debate than independence of irrelevant matches, but, on the basis of the motivation of $SC$ in Section~\ref{Sec32}, there exists an obvious way to soften it by being more tolerant in the case of opponents: $X_i$ is not required to be better than $X_j$ if it achieves the same result against stronger opponents.

\begin{axiom} \label{Axiom43}
\emph{Weak self-consistency} ($WSC$):
A scoring procedure $f: \mathcal{R}^n \to \mathbb{R}^n$ is called \emph{weakly self-consistent} if the following implication holds for any ranking problem $(N,R,M) \in \mathcal{R}^n$ and for any objects $X_i,X_j \in N$:
if there exists a decomposition of the ranking problem $(N,R,M)$ into $m$ unweighted ranking problems -- that is, $R = \sum_{p=1}^m R^{(p)}$, $M = \sum_{p=1}^m M^{(p)}$, and $(N,R^{(p)},M^{(p)}) \in \mathcal{R}_U^n$ is an unweighted ranking problem for all $p = 1,2, \dots ,m$ -- in a way that enables a one-to-one mapping $g^{(p)}$ from $O^{(p)}_i$ onto $O^{(p)}_j$ such that $r_{ik}^{(p)} \geq r_{j \mathfrak{g}^{(p)}(k)}^{(p)}$ and $f_k(N,R,M) \geq f_{\mathfrak{g}^{(p)}(k)}(N,R,M)$ for all $p = 1,2, \dots ,m$ and $X_k \in O_i^{(p)}$, then
$f_i(N,R,M) \geq f_{j}(N,R,M)$, furthermore, $f_i(N,R,M) > f_{j}(N,R,M)$ if $r_{ik}^{(p)} > r_{j \mathfrak{g}^{(p)}(k)}^{(p)}$ for at least one $1 \leq p \leq m$ and $X_k \in O_i^{(p)}$.
\end{axiom}

It can be seen that self-consistency (Axiom~\ref{Axiom32}) formalizes conditions \ref{SC_con1}-\ref{SC_con6}, while weak self-consistency only requires the scoring procedure to satisfy \ref{SC_con1}, \ref{SC_con2}, and \ref{SC_con4}-\ref{SC_con6}.

\begin{corollary} \label{Col45}
$SC$ implies $WSC$.
\end{corollary}

\begin{lemma} \label{Lemma42}
The row sum method is weakly self-consistent.
\end{lemma}

\begin{proof}
Let $(N,R,M) \in \mathcal{R}^n$ be a ranking problem such that $R = \sum_{p=1}^m R^{(p)}$, $M = \sum_{p=1}^m M^{(p)}$ and $(N,R^{(p)},M^{(p)}) \in \mathcal{R}_U^n$ is an unweighted ranking problem for all $p = 1,2, \dots ,m$.
Let $X_i,X_j \in N$ be two objects and assume that for all $p = 1,2, \dots ,m$ there exists a one-to-one mapping $g^{(p)}$ from $O^{(p)}_i$ onto $O^{(p)}_j$, where $r_{ik}^{(p)} \geq r_{j \mathfrak{g}^{(p)}(k)}^{(p)}$ and $s_k(N,R,M) \geq s_{\mathfrak{g}^{(p)}(k)}(N,R,M)$.

Obviously, $s_i(N,R,M) = \sum_{p=1}^m \sum_{X_k \in O_i^{(p)}} r_{ik} \geq \sum_{p=1}^m \sum_{X_k \in O_j^{(p)}} r_{j \mathfrak{g}^{(p)}(k)} = s_j(N,R,M)$. Furthermore, $s_i(N,R,M) > s_j(N,R,M)$ if $r_{ik}^{(p)} > r_{j \mathfrak{g}^{(p)}(k)}^{(p)}$ for at least one $p = 1,2, \dots ,m$.
\end{proof}

The last possibility result comes immediately. 

\begin{proposition} \label{Prop45}
The row sum method is independent of irrelevant matches and weakly self-consistent.
\end{proposition}

\begin{proof}
It follows from Lemmata~\ref{Lemma31} and \ref{Lemma42}.
\end{proof}

According to Lemma~\ref{Lemma42}, the violation of self-consistency by row sum (see Corollary~\ref{Col31}) is a consequence of condition~\ref{SC_con3}: the row sums of $X_i$ and $X_j$ are the same even if  $X_i$ achieves the same result as $X_j$ against stronger opponents.

It is a crucial argument against the use of row sum for ranking in tournaments which are not organized in a round-robin format, supporting the empirical findings of \citet{Csato2017c} for Swiss-system chess team tournaments.

\section{Conclusions} \label{Sec5}

\begin{table}[htbp]
\centering
\caption{Summary of the axioms}
\label{Table1}
\begin{subtable}{\textwidth}
\centering
    \begin{tabularx}{0.9\textwidth}{l CC} \toprule
    Axiom & Abbreviation & Definition \\ \midrule
    Independence of irrelevant matches & $IIM$ & Axiom~\ref{Axiom31} \\
    Self-consistency & $SC$ & Axiom~\ref{Axiom32} \\
    Macrovertex independence & $MVI$ & Axiom~\ref{Axiom41} \\
    Macrovertex autonomy & $MVA$ & Axiom~\ref{Axiom42} \\
    Weak self-consistency & $WSC$ & Axiom~\ref{Axiom43} \\ \bottomrule
    \end{tabularx}
\end{subtable}

\vspace{0.25cm}
\begin{subtable}{\textwidth}
    \begin{tabularx}{\textwidth}{l CCC} \toprule
          & \multicolumn{3}{c}{Is it satisfied by the particular method?} \\
    Axiom & Row sum (Definition~\ref{Def21}) & Generalized row sum (Definition~\ref{Def22}) & Least squares (Definition~\ref{Def23}) \\ \midrule
    Independence of irrelevant matches & \textcolor{PineGreen}{\ding{52}} & \textcolor{BrickRed}{\ding{55}} & \textcolor{BrickRed}{\ding{55}} \\
    Self-consistency & \textcolor{BrickRed}{\ding{55}} & \textcolor{PineGreen}{\ding{52}} & \textcolor{PineGreen}{\ding{52}} \\
    Macrovertex independence & \textcolor{PineGreen}{\ding{52}} & \textcolor{PineGreen}{\ding{52}} & \textcolor{PineGreen}{\ding{52}} \\
    Macrovertex autonomy & \textcolor{PineGreen}{\ding{52}} & \textcolor{PineGreen}{\ding{52}} & \textcolor{PineGreen}{\ding{52}} \\
    Weak self-consistency & \textcolor{PineGreen}{\ding{52}} & \textcolor{PineGreen}{\ding{52}} & \textcolor{PineGreen}{\ding{52}} \\ \bottomrule
    \end{tabularx}
\end{subtable}
\end{table}

The paper has discussed the problem of ranking objects in a paired comparison-based setting, which allows for different preference intensities as well as incomplete and multiple comparisons, from a theoretical perspective. We have used five axioms for this purpose, and have analysed three scoring procedures with respect to them. Our findings are presented in Table~\ref{Table1}.

However, our main contribution is a basic impossibility result (Theorem~\ref{Theo31}). The theorem involves two axioms, one -- called independence of irrelevant matches -- posing a kind of independence concerning the order of two objects, and the other -- self-consistency -- requiring to rank objects with an obviously better performance higher.

We have also aspired to get some positive results. Domain restriction is fruitful in the case of round-robin tournaments (Proposition~\ref{Prop43}), whereas limiting the intensity and the number of preferences does not eliminate impossibility if the number of objects is meaningful (Proposition~\ref{Prop42}, but Proposition~\ref{Prop41}). Self-consistency has a natural weakening, satisfied by row sum besides independence of irrelevant matches (Proposition~\ref{Prop45}), although $SC$ seems to be the more plausible property than $IIM$.
Independence of irrelevant matches can be refined through the concept of macrovertex such that the relative ranking of two objects should not depend on an outside comparison only if the comparison multigraph have a special structure. The implied possibility theorem (Proposition~\ref{Prop44}) is more general than the positive result in the case of round-robin ranking problems (consider Corollary~\ref{Col44}).

There remains an unexplored gap between our impossibility and possibility theorems since the latter allows for more than one scoring procedure. Actually, generalized row sum and least squares methods cannot be distinguished with respect to the properties examined here, as illustrated by Table~\ref{Table1}.\footnote{~Some of their differences are highlighted by \citet{Gonzalez-DiazHendrickxLohmann2013}.}
The loss of independence of irrelevant matches makes characterizations on the general domain complicated since self-consistency is not an axiom easy to seize. Despite these challenges, axiomatic construction of scoring procedures means a natural continuation of the current research.

\section*{Acknowledgements}
\addcontentsline{toc}{section}{Acknowledgements}
\noindent
We thank \emph{S\'andor Boz\'oki} for useful advice. \\
Anonymous reviewers provided valuable comments and suggestions on earlier drafts. \\
The research was supported by OTKA grant K 111797 and by the MTA Premium Post Doctorate Research Program.


\end{document}